\documentclass[10pt, conference, letterpaper]{IEEEtran}
\IEEEoverridecommandlockouts
\usepackage{amsmath}
\usepackage{amsthm}
\usepackage{amsfonts}
\usepackage{amssymb}
\usepackage{graphicx}
\usepackage{xcolor}
\usepackage{subcaption}
\usepackage{comment}
\usepackage{multirow}
\usepackage{verbatim}
\usepackage{dsfont}
\usepackage[ruled, vlined, linesnumbered]{algorithm2e}
\usepackage{tikz}
\newcommand*\circled[1]{\tikz[baseline=(char.base)]{
            \node[shape=circle,draw,inner sep=1pt] (char) {#1};}}

\newtheorem{lemma}{Lemma}
\newtheorem{theorem}{Theorem}

\def\eg{\textit{e.g.}}
\def\ie{\textit{i.e.}}
\def\etal{\textit{et al.}}

\SetKw{Return}{return}
\SetKw{Break}{break}
\SetKw{Continue}{continue}

\def\heavyedge{\textsf{Heavy-Edge}}
\def\ouralg{\textsf{A-SRPT}}

\def\opt{$\textit{OPT}_A$}
\def\optAone{$\textit{OPT}_{{A_{1}}}$}
\def\optAonepred{$\textit{OPT}_{\tilde{A}_{1}}$}
\def\aux{$\textit{AUX}$}

\begin{document}
\title{Prediction-Assisted Online Distributed Deep Learning Workload Scheduling in GPU Clusters}


\author{\IEEEauthorblockN{Ziyue Luo\IEEEauthorrefmark{1},
Jia Liu\IEEEauthorrefmark{1},
Myungjin Lee\IEEEauthorrefmark{2},
Ness B. Shroff\IEEEauthorrefmark{1}}
\IEEEauthorblockA{\IEEEauthorrefmark{1} Dept. of ECE, The Ohio State University, USA, Email: luo.1457@osu.edu, \{liu, shroff\}@ece.osu.edu}
\IEEEauthorblockA{\IEEEauthorrefmark{2} Cisco Research, USA, Email: myungjle@cisco.com}
\thanks{This work has been supported in part by Cisco Research Award PO-USA000EP312336, by NSF grants CAREER CNS-2110259, IIS-2324052, NSF AI Institute (AI-EDGE) CNS-2112471, CNS-2312836, CNS-2106933, CNS-2106932, CNS-1955535, and CNS-1901057, by DARPA YFA D24AP00265, by ONR grant N00014-24-1-2729, by AFRL grant PGSC-SC-111374-19s, by Army Research Office under Grants W911NF-21-1-0244 and W911NF-24-2-0205, and was sponsored by the Army Research Laboratory and was accomplished under Cooperative Agreement Number W911NF-23-2-0225. The views and conclusions contained in this document are those of the authors and should not be interpreted as representing the official policies, either expressed or implied, of the Army Research Laboratory or the U.S. Government. The U.S. Government is authorized to reproduce and distribute reprints for Government purposes notwithstanding any copyright notation herein.}
}

\maketitle

\begin{abstract}
The recent explosive growth of deep learning (DL) models has necessitated a compelling need for efficient job scheduling for distributed deep learning training with mixed parallelisms (DDLwMP) in GPU clusters. This paper proposes an adaptive shortest-remaining-processing-time-first (\textsf{A-SRPT}) scheduling algorithm, a novel prediction-assisted online scheduling approach designed to mitigate the challenges associated with DL cluster scheduling. By modeling each job as a graph corresponding to heterogeneous Deep Neural Network (DNN) models and their associated distributed training configurations, \textsf{A-SRPT} strategically assigns jobs to the available GPUs, thereby minimizing inter-server communication overhead. Observing that most DDLwMP jobs recur, \textsf{A-SRPT} incorporates a random forest regression model to predict training iterations. Crucially, \textsf{A-SRPT} maps the complex scheduling problem into a single-machine instance, which is addressed optimally by a preemptive ``shortest-remaining-processing-time-first'' strategy. This optimized solution serves as a guide for actual job scheduling within the GPU clusters, leading to a theoretically provable competitive scheduling efficiency. 
We conduct extensive real-world testbed and simulation experiments to verify our proposed algorithms.
\end{abstract}

\section{Introduction}
\label{sec::intro}

Distributed deep learning (DDL) has recently achieved remarkable successes across multiple domains, \eg, natural language processing (NLP)~\cite{brown2020language}, computer vision~\cite{he2016deep}, and computer networks~\cite{chen2018auto}.
However, the training of deep neural network (DNN) models is compute-intensive, requiring dedicated, powerful, and expensive GPU clusters~\cite{xiao2018gandiva,qiao2021pollux,weng2022mlaas},
This has necessitated developing algorithms to efficiently schedule distributed deep learning training jobs with mixed parallelisms (DDLwMP), including but not limited to data parallelism~\cite{li2014scaling}, model parallelism~\cite{shoeybi2019megatron} and pipeline parallelism~\cite{harlap2018pipedream}. 
Such scheduling algorithms are pivotal for resource allocations in GPU clusters to orchestrate DDLwMP jobs' execution.

In the areas of DDL scheduling algorithm design, many early attempts adopted a preemptive scheduling approach that permits pausing, resumption, and reallocation of running jobs for better flexibility. 
However, with ever-increasing learning model sizes, interrupting DDL job executions, including saving/loading training models into/from the host memory and potentially reallocating jobs to a different set of GPUs, incurs large overhead on the order of seconds to minutes~\cite{qiao2021pollux}. 
To pursue improved resource utilization and consistent processing of DDL jobs, some recent studies have shifted their focus towards designing non-preemptive ML cluster scheduling algorithms~\cite{han2020scheduling,wang2020communication,yu2022scheduling}, where the scheduler dedicates a set of GPUs solely for each DDLwMP job to ensure that all allocated GPUs execute simultaneously without interruption until the job's completion. 
However, all aforementioned works are designed for DDL jobs {\em without} mixed parallelisms.
To date, designing scheduling algorithms for DDLwMP remains in its infancy and there are several highly non-trivial challenges:

\textit{1)}~DDLwMP jobs differ significantly in their model architectures, consisting of diverse types of DNN layers. 
The mixture of parallelisms results in complex computation and communication patterns during training.
Thus, optimally placing DDLwMP jobs across the available GPUs, taking into account their model architectures and parallel paradigms, 
is highly challenging.
Further, resource fragmentation (available GPUs are scattered across partially occupied servers due to frequent small job allocations) exacerbates the problem.

\textit{2)}~The unpredictability of future workloads introduces another challenge, rendering the scheduling task an online problem. 
Due to the non-preemption constraints, greedily scheduling existing jobs to fully occupy the cluster's computational resources can lead to fragmentation issues and significantly delay incoming jobs, thus increasing overall latencies. 
Thus, strategic orchestration of the available jobs is needed to minimize the total job completion time: the algorithm should schedule sufficiently many jobs to maximize resource utilization 
{while reserving resources for future job arrivals.}

\textit{3)}~Many existing non-preemptive scheduling designs require the knowledge of training iterations upon jobs' submissions to estimate job training durations. 
However, DNN model training is a feedback-dependent exploration process~\cite{karmaker2021automl}. 
It is common for users to submit multiple jobs exploring different configurations of hyper-parameters and terminate most jobs due to random errors or sub-optimal convergence performance~\cite{weng2022mlaas, gu2019tiresias}. 
This implies that the actual number of job training iterations is {\em uncertain}. 
Blindly scheduling jobs according to the user-specified training iterations could lead to suboptimal performance.

To address these challenges, in this paper, we propose an adaptive shortest-remaining-processing-time-first (\ouralg{}) scheduling algorithm. 
{Our design contains two key components: 1) a {\bf GPU mapping} algorithm that judiciously assigns a DDLwMP job to a specific set of GPUs, thereby minimizing the data communication overhead during job training; and 2) a \textbf{prediction-assisted online scheduling} algorithm that strategically schedules DDLwMP jobs by incorporating a job total training iteration prediction model.
}
Our main contributions and key results are summarized as follows:

\begin{list}{\labelitemi}{\leftmargin=1em \itemindent=-0.5em \itemsep=.2em}

\item 
We represent DDLwMP jobs with various models and distributed training configurations as graphs,
based on which we further develop the \heavyedge{} algorithm, a graph-cut-based method designed to strategically allocate each job to available GPUs across servers. 
\heavyedge{} emphasizes maximizing the use of high-bandwidth interconnection for GPUs within a server (\eg, NVLink~\cite{nvlink}), thereby improving overall training efficiency. 


\item {We tackle the uncertain training duration challenge by leveraging the recurrence of DDLwMP jobs. 
First, we use a random forest regression method~\cite{breiman2001random} to predict training iterations from historical job execution traces.}
Then, by leveraging this prediction model, we develop a prediction-assisted online scheduling framework called \ouralg{} based on a {\em two-step} approach: 
1) We show that the original complex multi-dimensional GPU clustering problem can be simplified as a preemptive {\em single-machine} scheduling problem with the predicted number of training iterations for each DDLwMP job.
This simplification enables the use of the shortest remaining processing time (SRPT) principle~\cite{lawler1993sequencing}, which is {\em optimal} in scheduling jobs in the hypothetical single-machine problem; 
2) We use the virtual single-machine SRPT solution to guide our non-preemptive scheduling decisions for DDLwMP jobs in the actual cluster. 
This two-step approach allows us to design DDLwMP scheduling schemes with theoretical performance guarantee. 

\item {To validate the effectiveness of our proposed designs, we conduct {\em real-world} trace-driven testbed experiments and simulation studies based on profiled DDL workloads with mixed DNN models and a two-month DL workload trace~\cite{weng2022mlaas}.}
Our experimental results verify the superiority of our proposed algorithms over state-of-the-art DDL scheduling algorithms. 
Specifically, our proposed algorithm outperforms all baseline designs and achieves up to 92\% total job completion time reduction.
\end{list}


\section{Background and Related Work}
\label{sec::related_work}

\noindent{\bf 1) Parallelisms for Distributed DNN Training:}
DNN training is an iterative process to minimize a loss function~\cite{goodfellow2016deep}, where each iteration consists of forward propagation (FP), backward propagation (BP), and gradient update, all of which are based on mini-batches. 
The advent of large DNN models has driven the development of distributed DNN training to speed up DNN training. 
To enable distributed DNN training, data~\cite{li2014scaling}, model~\cite{shoeybi2019megatron}, and pipeline parallelisms~\cite{harlap2018pipedream, huang2019gpipe, luo2022efficient}, as shown in Fig.~\ref{fig:parallel}, are the most common.

Data parallelism {(Fig.\ref{fig:parallel}(a))} trains mini-batches on different GPUs, followed by gradient synchronization using ring AllReduce (RAR)~\cite{sergeev2018horovod} or tree AllReduce (TAR)~\cite{sanders2009two}. 
RAR forms a logical ring for communication~\cite{yu2022gadget}, while TAR uses double binary trees~\cite{sanders2009two} (e.g., NVIDIA NCCL~\cite{nccl}). 
This method requires each GPU to host a full DNN model, limiting it to small-size models. 
Model parallelism {(Fig.~\ref{fig:parallel}(b))} trains large models by distributing FPs and BPs across GPUs, each hosting a different model stage.
However, model parallelism suffers from low utilization as only one GPU is active at a time.

Building on model parallelism, pipeline parallelism {(Fig.\ref{fig:parallel}(c))} sequentially injects mini-batches into the system to allow simultaneous GPU processing.
Each model stage can have multiple replicas~\cite{fan2021dapple,luo2022efficient} trained with data parallelism to reduce stage processing time.
Pipeline parallelism can be further divided into asynchronous and synchronous pipelines.
Synchronous pipeline~\cite{huang2019gpipe, luo2022efficient} maintains a synchronization barrier between training iterations, enforcing synchronous gradient updates across all model stages to achieve a better convergence performance. 
However, such synchronization barriers may interrupt the pipeline and delay new mini-batch entries, leading to low GPU utilization. 
Asynchronous pipeline~\cite{narayanan2021memory} improves GPU utilization by continuously injecting mini-batches to increase training throughput at the price of (slight) model convergence degradation~\cite{harlap2018pipedream}.
In this work, we consider asynchronous pipeline due to its higher training efficiency.


\begin{figure}[t!]
  \centering
  \includegraphics[width=\columnwidth]{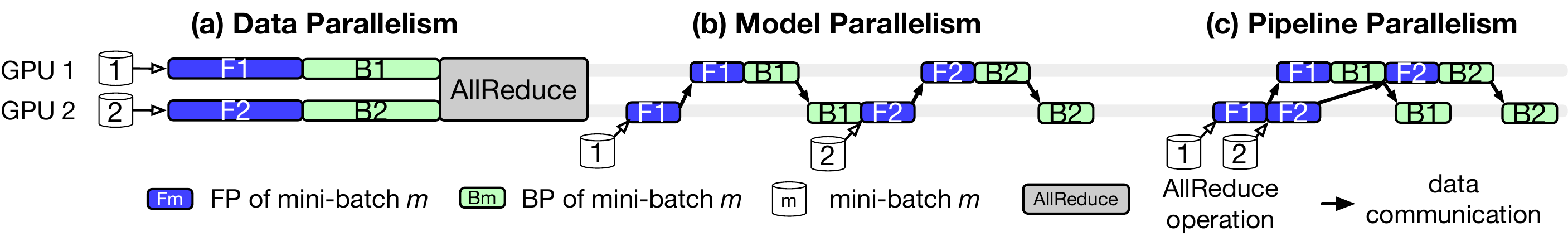}
  \caption{Three typical parallelisms for distributed DNN training.}
  \label{fig:parallel}
\end{figure}
\setlength{\textfloatsep}{0pt}

\noindent{\bf 2) Online DDL Job Scheduling:}
Early attempts on online DDL job scheduling focused on preemptive algorithms.
For instance, Optimus~\cite{peng2018optimus} constructs resource-performance models for dynamic GPU scaling to minimize completion time of data-parallel jobs. 
Gandiva~\cite{xiao2018gandiva} uses scaling heuristics for GPU-sharing across multiple jobs. 
GADGET~\cite{yu2022gadget} balances communication overhead and contention for resource scheduling for RAR jobs. 
Tiresias~\cite{gu2019tiresias} prioritizes jobs based on training duration metrics. 
Pollux~\cite{qiao2021pollux} adapts resources to optimize \textit{good-put}, a metric combining throughput and statistical efficiency.
Non-preemptive scheduling research is more limited. 
SPIN~\cite{han2020scheduling} focuses on minimizing makespan for placement-sensitive jobs. 
An online framework in~\cite{wang2020communication} addresses communication contention among DDL jobs. 
An offline approximation algorithm in~\cite{yu2022scheduling} tackles communication overhead and network contention for RAR jobs.
However, all existing methods above only considered a {\em single} parallelism.
By stark contrast, in this work, we propose a {\em non-preemptive} online scheduling algorithm for DDLwMP DDL training jobs 
with theoretical performance guarantees.

It is worth noting that most previous DDL job scheduling works rely on the knowledge of job training duration/iterations (some using predictive techniques based on historical runtimes~\cite{han2020scheduling,gu2019tiresias,peng2018optimus}). 
Abdullah \etal~\cite{faisal2024will} proposed to enhance ML job completion predictability using weighted-fair-queueing for bounded preemption.
However, prioritizing jobs by predicted execution time can lead to inaccurate GPU allocation and long wait times for short jobs.  
Inspired by recent advances in learning-augmented online preemptive scheduling for single machine~\cite{bampis2022scheduling}, we propose an online prediction-assisted algorithm for non-preemptive DDLwMP job scheduling to delay long jobs to expedite shorter ones.

\section{System Model}
\label{sec::sys}

We consider a GPU cluster consisting of $M$ inter-connected homogeneous GPU servers. Each server $m\in [M]$\footnote{We use $[X]$ to denote the set $\{1,2,\ldots, X\}$.} is equipped with $g$ GPUs, yielding a total of $G=Mg$ GPUs within the cluster. The bidirectional (\ie, incoming and outgoing) NIC bandwidth on each machine is denoted as $B_{\texttt{inter}}$.
The intra-server bidirectional GPU communication bandwidth (\eg, PCIe, and NVLink~\cite{nvlink}) is denoted as $B_{\texttt{intra}}$, which is typically one to two orders of magnitude greater than $B_{\texttt{inter}}$.
The system works in a time-slotted fashion, over a potentially large span of $T$ time-slots.
There are $I$ DDLwMP jobs in total in the cluster, and job $i\in[I]$ is submitted at time $r_i\in [T]$.
We note that our proposed scheduling algorithm for DDLwMP jobs also includes single-GPU jobs as a special case, thus offering general support for all DDL workloads.
In what follows, we zoom into two key components in our system modeling.

\subsection{Workload Scheduling for DDL Jobs in GPU Cluster}

In our DDLwMP training setting, each job $i \in [I]$ requests to train a DNN model $\mathcal{D}_i$ for $n_i$ iterations using a specific distributed configuration. 
$\mathcal{D}_i$ is divided into $S_i$ stages, each of which consists of some consecutive DNN layers. 
For improved training efficiency, stages can further be replicated across multiple GPUs in a data-parallel fashion~\cite{harlap2018pipedream, fan2021dapple}, allowing varying degrees of data-parallelism across different stages.
The processing of a single mini-batch by a stage is distributed over the GPUs.
Let $k_{i,s}$ denote the number of data-parallel replicas for stage $s \in [S_i]$ of job $i$, which equals the required GPUs for this stage. 
Thus, the total GPUs needed for job $i$ is $g_i = \sum_{s \in [S_i]} k_{i,s}$. 
A single-GPU job is a special case with one non-replicated stage.
Our distributed training configuration covers the following parallelisms as special cases: 1) data parallelism (single-stage, multiple replicas), 2) model parallelism (multiple non-replicated stages), and 3) pipeline parallelism (other cases). We assume parallel configurations are given through pipeline planning~\cite{tarnawski2021piper, luo2022efficient}.


On a given GPU, the time required for the FP (resp. BP) of a mini-batch over a replica of stage $s$ for job $i$ is denoted by $p^f_{i,s}$ (resp. $p^b_{i,s}$). 
The incoming and outgoing data size (\ie, activations during FP and gradients during BP) for each training iteration per replica of stage $s$ in job $i$ are denoted by $d_{i,s}^{in}$ and $d_{i,s}^{out}$ respectively. 
We use $h_{i,s}$ to represent the size of trainable parameters for job $i$ and stage $s$.

We use $x^m_{i,s}$ to represent the number of GPUs allocated on server $m$ to host stage $s$ of job $i$, and use $t_i$ to denote the starting time of job $i$. 
Accordingly, an amount of $x^m_{i,s}/g$ bandwidth for the stage is reserved at the incoming and outgoing NIC.
Let $\alpha_i(\{x^m_{i,s}\})$ represent the per-iteration training time of job $i$ given its GPU allocation $\{x^m_{i,s}\}$, which will often be simplified as $\alpha_i$ for notational simplicity henceforth if no confusion arises from the context. 
The characterization of $\alpha_i$ will be specified later in this section.
To ensure schedule feasibility, we have the following constraints:

\vspace{-5mm}
{\small
\begin{align}
& t_i \geq r_i, \forall i\in [I], \label{eqn:start_time}\\
& \sum\limits_{m\in [M]}x^m_{i,s} = k_{i,s}, \forall i\in [I], s\in [S_i], \label{eqn:feasible_placement}\\
& \sum\limits_{i\in [I]: t_i \leq t \leq t_i + n_i \alpha_i}\sum\limits_{s\in [S_i]}x^m_{i,s} \leq g, \forall m\in [M], t\in[T].\label{eqn:server_capacity}
\end{align}
}

Here, Constraint~(\ref{eqn:start_time}) ensures that each job is scheduled to start only after its submission; Constraint~(\ref{eqn:feasible_placement}) implies that all stage replicas of job $i$ are allocated in the cluster;
and Constraint~(\ref{eqn:server_capacity}) guarantees that the allocated GPUs for active jobs do not exceed each server’s capacity limit.

\subsection{Characterization of Per-Iteration Training Time $\alpha_i$}

As mentioned in Section~\ref{sec::related_work}, we focus on the widely adopted asynchronous pipeline parallel training~\cite{harlap2018pipedream,narayanan2021memory}.
We note that our design can be straightforwardly extended to synchronous pipeline parallelism~\cite{huang2019gpipe} by following the analytic model proposed in~\cite{luo2022efficient} for $\alpha_i$.
{Under asynchronous pipeline parallelism, as the execution of all stages is fully pipelined, the per-iteration training time is the maximum per-stage computation-communication time of a single stage (\ie, the bottleneck stage)~\cite{harlap2018pipedream, tarnawski2021piper}.}
{We use $\beta^m_{i,s}$ to denote the per-iteration training time of stage $s$ of job $i$ on machine $m$, which consists of the computation time for the current batch of the stage replicas on server $m$ in one iteration (denoted as $comp^m_{i,s}$), the data communication time for sending activations and gradients of the current batch into and out of the stage (denoted as $comm^m_{i,s}$), and the communication costs for synchronizing parameters among all the stage replicas using AllReduce operations ($AllReduce^m_{i,s}$).} 
The communication time (including both the FP and the BP) can be calculated as follows:

\vspace{-3mm}
{\small
\begin{equation}
	comp^m_{i,s} = 
	\begin{cases}
			p^f_{i,s} + p^b_{i,s}, & x^m_{i,s} > 0,\\
			0, &  x^m_{i,s} = 0.
	\end{cases}
\end{equation}
}
\vspace{-3mm}

To compute the inter-stage communication time when stage $s-1$ and/or $s$ are replicated over multiple GPUs, we evenly distribute the data being transmitted across inter-stage links. 
Thus, the per-iteration data communication time between each replica of stage $s-1$ and $s$ is $\frac{2d^{out}_{i,s-1}}{k_{i,s}}= \frac{2d^{in}_{i,s}}{k_{i,s-1}}$~\cite{luo2022efficient}. 
Hence, for stage $s\in [2, 3, \ldots, S_i-1]$, if $x^m_{i,s} > 0$, we have:

\vspace{-3mm}
{\small
\begin{multline}
	comm^m_{i,s} \!=\! \frac{(2d^{in}_{i,s}\frac{k_{i,s-1} - x^m_{i,s-1}}{k_{i,s-1}} + 2d^{out}_{i,s}\frac{k_{i,s+1} - x^m_{i,s+1}}{k_{i,s+1}})x^m_{i,s}}{(x^m_{i,s}/g) B_{\texttt{inter}}} \\ + \frac{2d^{in}_{i,s}\frac{x^m_{i,s-1}}{k_{i,s-1}} + 2d^{out}_{i,s}\frac{x^m_{i,s+1}}{k_{i,s+1}}}{B_{\texttt{intra}}},
\end{multline}
}

and $comm^m_{i,s}=0$ otherwise. The term $comm^m_{i,s}$ for the first and last stages can be calculated similarly.
{The data size communicated for each stage replica of stage $s$ in the AllReduce operation can be calculated as $\frac{2(k_{i,s}-1)}{k_{i,s}}h_{i,s}$~\cite{allreduce_time} for both RAR and TAR}, and the data communication time is bottlenecked by the minimum bandwidth between stage replicas. Hence, the time taken by the AllReduce operation for job $i$ stage $s$ is:

\vspace{-1mm}
{\small
\begin{equation}
	AllReduce^m_{i,s}= 
		\begin{cases}
			\frac{2(k_{i,s}-1)h_{i,s}}{k_{i,s}\frac{x^m_{i,s}}{g} B_{\texttt{inter}}}, & \text{if } x^m_{i,s} < k_{i,s}, \\
			\frac{2(k_{i,s}-1)h_{i,s}}{k_{i,s}B_{\texttt{intra}}}, & \text{if }  x^m_{i,s} = k_{i,s}.
		\end{cases}
\end{equation}
}

Here, in the first case, the bottleneck is due to the server NIC bandwidth, while in the second case, all data communication is conducted via the intra-server connection. 
{Lastly, by putting all things together and in line with existing formulations on pipeline scheduling~\cite{fan2021dapple,tarnawski2021piper,luo2022efficient}, we obtain the per-iteration training time $\alpha_i$ for processing a single batch as follows:}

\vspace{-.1in}
{\small
\begin{multline}
\alpha_i = \max_{m\in \mathcal{M}, s\in [S_i]}\beta^m_{i,s} \\
= \max_{m\in \mathcal{M}, s\in [S_i]} (comp^m_{i,s} + comm^m_{i,s} + AllReduce^m_{i,s}).
\label{eqn:periteration_time}
\end{multline}
}

{Additionally, some distributed communication engines (\eg, BytePS~\cite{peng2019generic}) enable strategic overlapping of AllReduce operations with backward computation. For example, gradients for layer $l$ can be synchronized using AllReduce while simultaneously computing gradients for layer $l-1$. To account for this overlapping, one can apply model-dependent coefficients to the backward computation time and AllReduce time~\cite{yu2021sum}.}

Let $\alpha_i^{\max}$ and $\alpha_i^{\min}$ denote the maximum and minimum per-iteration training times of job $i$ given a GPU assignment, respectively.
$\alpha_i^{\max}$ can be computed using Eq.~(\ref{eqn:periteration_time}) if the job is assigned to $g_i$ servers, with each server holding a single-stage replica and assigned a bandwidth of $1/g\times B_{\texttt{inter}}$.
However, evaluating $\alpha_i^{\min}$ for each job requires searching through an exponential number of possible GPU assignments, which is computationally intractable. To address this challenge, we will propose an estimation strategy to be described in Sec.~\ref{subsec:heavgedge}.

\subsection{The Online DDLwMP Job Scheduling Problem}

In this paper, our goal is to minimize the total DDLwMP job completion in a time horizon of length $T$, which can be evaluated as $\sum_{i\in [I]}(t_i + n_i\alpha_i)$. Putting all modeling together, we can formulate our DDLwMP job scheduling problem as:

\vspace{-.1in}
{\small
\begin{align}
\label{eqn:online_scheduling_problem}\text{ Minimize }
      & \sum\limits_{i\in [I]}(t_i + n_i\alpha_i) \\
    \text{subject to } 
    & \text{(\ref{eqn:start_time})--(\ref{eqn:server_capacity})} , x^m_{i,s} \in \mathbb{\mathbb{N}}, \forall m\in [M], i\in [I], s\in [S_i], \nonumber\\ 
    &t_i\in [T], \forall i\in [I]. \nonumber
\vspace{-.1in}
\end{align}
}
{We note that Problem~(\ref{eqn:online_scheduling_problem}) is an integer non-convex program due to the intricate modeling of the per-iteration training time $\alpha_i$.
Moreover, another key challenge in Problem~\eqref{eqn:online_scheduling_problem} stems from the uncertain job submission time $r_i$ and the unknown number of job training iterations $n_i$, which necessitates online algorithmic designs.
In fact, the offline variant of Problem~(\ref{eqn:online_scheduling_problem}), where $r_i$, $n_i$ and $\alpha_i$ are all predetermined (rendering the problem of scheduling parallelizable tasks~\cite{turek1992approximate}) is NP-hard.} 
To address these challenges, we propose a prediction-assisted algorithm for optimizing the online DDLwMP job scheduling in GPU clusters.
\section{Prediction-Assisted Online Job Scheduling Algorithm}
\label{sec::schedule}

\subsection{Basic Idea}

The complexities of Problem~(\ref{eqn:online_scheduling_problem}) arise from two distinct perspectives: 1) The sensitivity of DDLwMP jobs to GPU placement (the per-iteration training time, can significantly vary with different placements); and 2) the inherent online nature of the problem (not only are the job arrival times unknown, but the actual number of job execution iterations is also typically uncertain in practice).

To address these unique challenges, we introduce a new online scheduling algorithm named adaptive shortest-remaining-processing-time-first (\ouralg{}) to solve Problem~\eqref{eqn:online_scheduling_problem} based on the following key observations:
First, we note that the complex computation-communication structure of DDLwMP jobs can be effectively modeled using graphs. This realization leads us to develop a strategic graph partitioning algorithm called \heavyedge{}. This algorithm favors co-locating replicas with substantial communication requirements, thereby enhancing overall scheduling efficiency.

{Utilizing \heavyedge{} for job placement, we propose an online job scheduling framework for DDLwMP jobs with a theoretical competitive ratio guarantee. This framework is inspired by the proven optimality of preemptive Shortest Remaining Processing Time (SRPT) scheduling for jobs based on their predicted durations on a single machine~\cite{bampis2022scheduling}. In our approach, we construct a single-machine preemptive scheduling instance based on the original non-preemptive scheduling problem. This construction considers the size of each job and its predicted number of training iterations.} 

{We then apply SRPT to preemptively schedule these jobs within this hypothetical single-machine instance. 
The results obtained from this single-machine scheduling model are then used to guide the non-preemptive job allocation in the actual cluster environment.
In this way, jobs with larger predicted workloads are scheduled later, creating space for potentially future smaller jobs to be scheduled first, thus reducing the total job completion time.
}

\subsection{The \heavyedge{} GPU Mapping Algorithm} \label{subsec:heavgedge} 

In our \heavyedge{} GPU mapping algorithm, each job $i$ is assigned to a set of servers $\mathcal{M}_i$ for execution. 
Each server $m \in \mathcal{M}_i$ has $g_m$ available GPUs to host job $i$'s stage replicas, such that $\sum_{m\in \mathcal{M}_i}g_m = g_i$ ($g_m \leq g$ as some GPUs in the server may be occupied by existing jobs). 
We now map each stage replica of job $i$ to a GPU, with the goal to reduce inter-server communication to improve job training efficiency.


Toward this end, we model each job $i$ as a graph $\Omega = (\mathcal{V}, \mathcal{E})$, where vertices $\mathcal{V}$ represent stage replicas and edges $\mathcal{E}$ denote data communication, with edge weights indicating communication data size. 
For inter-stage communication between stages $s-1$ and $s$, we assign edges with weight $\frac{2d^{out}_{i,s-1}}{k_{i,s}}= \frac{2d^{in}_{i,s}}{k_{i,s-1}}$ for each replica pair. 
For intra-stage communication (AllReduce) in stage $s$, we handle RAR and TAR differently. 
In RAR, replicas form a ring with edges weighted $\frac{2(k_{i,s}-1)}{k_{i,s}}h_{i,s}$. 
For TAR, edges connect replica pairs linked in double binary trees, weighted $\frac{(k_{i,s}-1)}{k_{i,s}}h_{i,s}$, which is halved compared to RAR. 
This reduction is due to the structure of the double binary trees, where each tree processes half of the total data~\cite{nccl}.

\begin{figure}[!t]
  \centering
  \includegraphics[width=0.65\columnwidth]{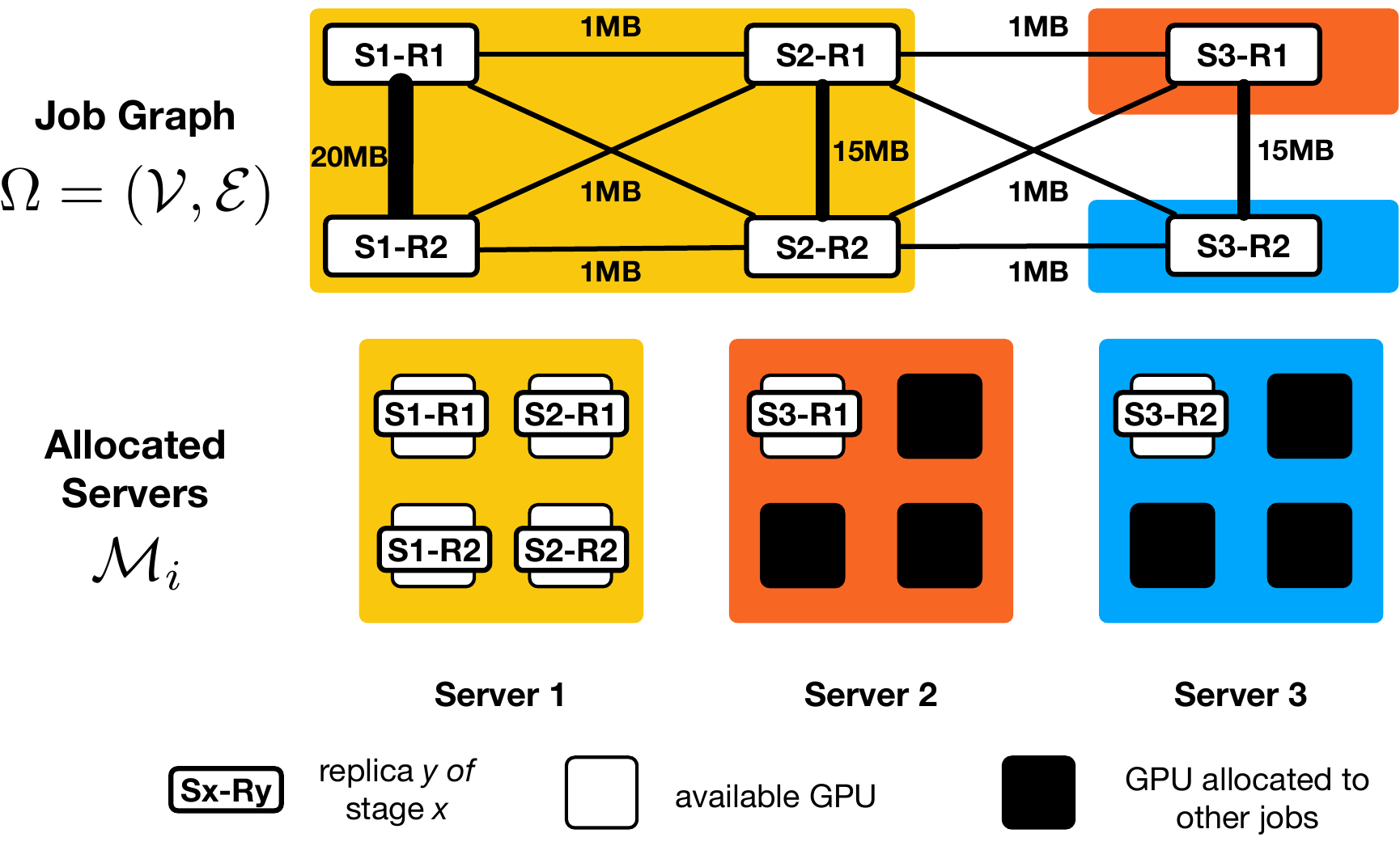}
  \caption{GPU mapping: An illustrative example.}
  \label{fig:gpu_mapping}
\end{figure}

As a result, the GPU mapping problem is equivalent to a graph cut problem that partitions a graph into $|\mathcal{M}_i|$ subgraphs of size $g_m$ to minimize inter-server communication (total cut weight among subgraphs) and maximize intra-server communication (total edge weights within subgraphs).
Fig.~\ref{fig:gpu_mapping} illustrates an example of GPU mapping. 
The job consists of three stages, each with two replicas. 
The job is assigned to three servers with four, one, and one available GPU(s), respectively. 
We partition the job graph into three subgraphs, each corresponding to the set of GPUs in a server with the same color.
Unfortunately, this graph partitioning problem is an NP-complete balanced graph cut problem~\cite{andreev2004balanced} even with equal GPUs per server, and not to mention with varying GPU availability. 
To address this challenge, we propose the \heavyedge{} approach, which greedily assigns heavily connected stage replicas to servers as follows.
In \heavyedge{}, we start by sorting the servers in $|\mathcal{M}_i|$ based on the available GPU numbers in a descending order, denoted as $\{m_1, m_2, \ldots, m_{\mathcal{M}_i}\}$.
Vertices in $\mathcal{V}$ (i.e., stage replicas) are then assigned to these servers from $m_1$ to $m_{\mathcal{M}_i}$.
We denote the current server for assignment as $m$ and use \texttt{node\_set} to denote the set of vertices assigned to $m$, which is initialized as $\emptyset$. 
Next, we consider two cases: 1) if $|\mathcal{V}|$ equals $m$'s GPU count, all replicas are assigned to it; 2) for single-GPU servers, we assign the vertex with the minimum total edge weight.
In the case of a server with multiple GPUs and there are remaining vertices, the GPU mapping process follows the ``\heavyedge{}'' principle: we iteratively add vertices to \texttt{node\_set} by finding the heaviest edge between assigned and unassigned vertices, prioritizing intra-server communication efficiency.
If no connecting edge exists, we randomly assign an unassigned vertex. This process continues until \texttt{node\_set} matches $m$'s number of available GPUs.

We use an example in Fig.~\ref{fig:gpu_mapping} to further illustrate our \heavyedge{} GPU mapping algorithm. The process begins by identifying the heaviest communication edge, \((\texttt{S1-R1}, \texttt{S1-R2})\), with a data size of 20MB, and assigning these nodes to \texttt{node\_set}, \ie, the first server. To optimize communication efficiency, we then allocate unassigned nodes directly connected to this pair (\ie, \texttt{S2-R1} and \texttt{S2-R2}), each with a 1MB connection to \texttt{S1-R1} and \texttt{S1-R2} respectively, to the same server, maximizing intra-server communication. This process continues sequentially for subsequent servers until all nodes are assigned, effectively minimizing inter-server communication overhead.


With the \heavyedge{} GPU mapping algorithm, we obtain the minimum achievable per-iteration training time $\tilde{\alpha}_i^{\min}$ for jobs, helping predict job training times. 
To minimize per-iteration time, each job is allocated to the fewest servers possible, utilizing the maximum number of interconnected high-bandwidth GPUs. 
For job $i$, a set of machines $\mathcal{M}_i$ is assigned, where servers $m_1$ to $m_{|\mathcal{M}_i|-1}$ contribute all $g$ GPUs, and the last server $m_{|\mathcal{M}_i|}$ contributes $g' \leq g$ GPUs. \heavyedge{} determines the GPU mapping, and $\tilde{\alpha}_i^{\min}$ is estimated using~(\ref{eqn:periteration_time}).

\subsection{The A-SRPT Online DDLwMP Job Scheduling Algorithm}

\textbf{1) Adaptive Shortest-Remaining-Processing-Time-First:}
{Our online scheduling algorithm is inspired by the online SRPT framework proposed in~\cite{chekuri2001approximation}, which is optimal for online scheduling for single-machine jobs with known durations over parallel machines. 
However, our problem is far more complex due to two critical aspects: 
1) Each job in our setting can span {\em multiple} GPUs, inducing complex inter-job communication patterns; 
2) The actual number of training iterations of jobs in our setting becomes known only upon job completion.}
Assume that we have a prediction model that predicts the number of training iterations $\tilde{n}_i$ for each training job $i$. 
We define the prediction error for job $i$, denoted by $\epsilon_i$, as the total absolute difference between the predicted and actual numbers of training iterations, \ie, $|n_i - \tilde{n}_i|$. 
Let $\epsilon$ and $\bar{\epsilon}$ denote the total prediction and average prediction errors, respectively, which can be computed as:

\vspace{-2mm}
{\small
\begin{equation}
	\epsilon = \sum\limits_{i\in[I]}\epsilon_i = \sum\limits_{i\in[I]}|n_i - \tilde{n}_i|, \quad \text{and} \quad \bar{\epsilon} = \frac{\epsilon}{I}.
	\label{eqn:error}
\end{equation}}
\vspace{-2mm}

{Our proposed design adopts the Shortest Remaining Processing Time (SRPT) strategy, which prioritizes available jobs with the least processing time. This approach is known to be delay-optimal in single-machine preemptive settings~\cite{lawler1993sequencing} and has been proven competitive even when job processing times are unknown until completion but can be estimated~\cite{bampis2022scheduling}.}

\begin{figure}[t!]
  \centering
  \includegraphics[width=0.8\columnwidth]{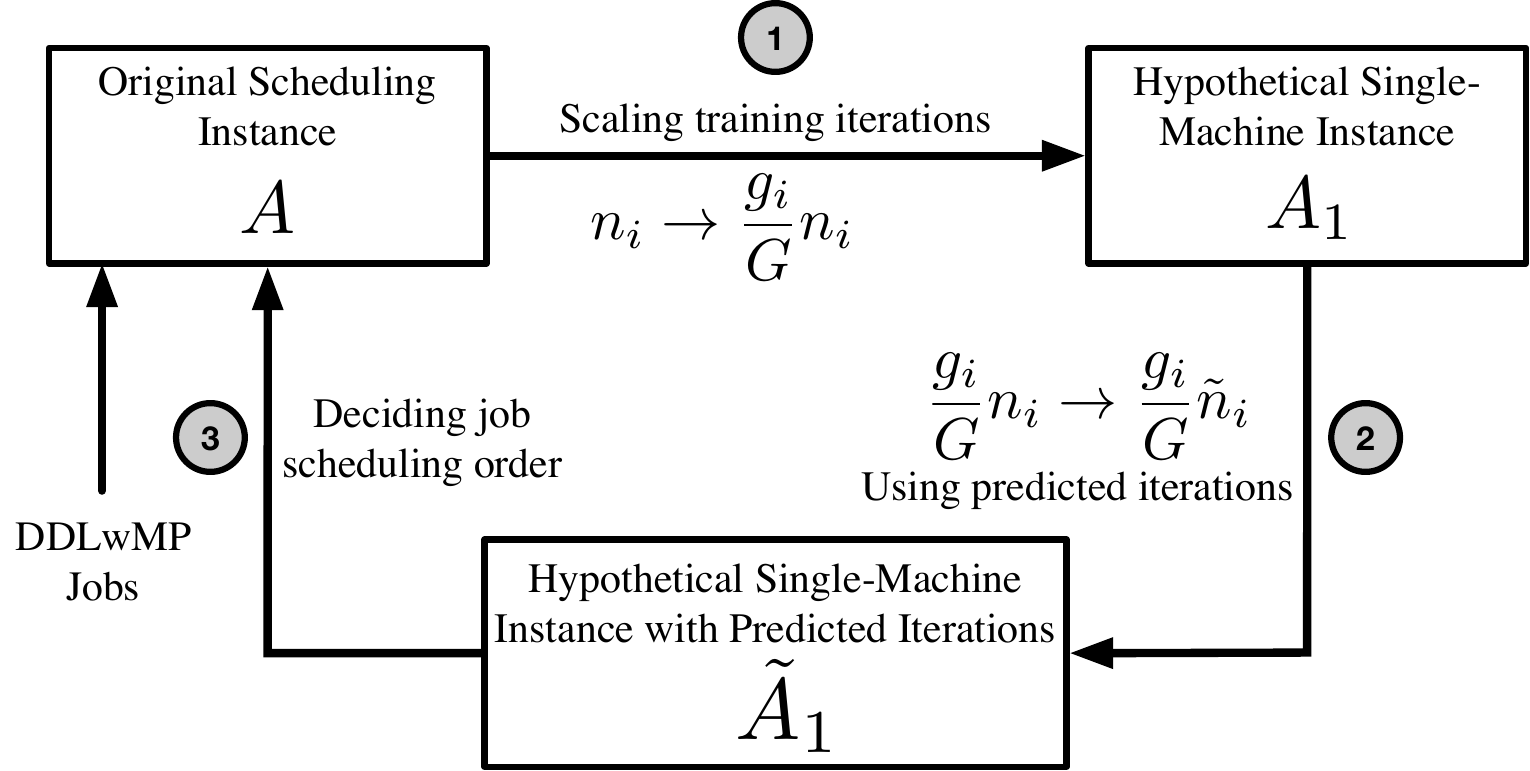}
  \caption{Algorithmic idea overview.}
  \label{fig:alg_overview}
\end{figure}

{We present an overview of our algorithmic idea in Fig.~\ref{fig:alg_overview}.}
{
We ``virtualize'' the entire GPU cluster as a `single machine' and proportionally scale down each job's workload (\circled{1}). 
Specifically, let instance $A$ denote the original online DDLwMP scheduling problem. 
We then define a new hypothetical single-machine preemptive online scheduling problem $A_1$, sharing $A$'s job set. 
In $A_1$, the number of training iterations for job $i$ is scaled to $\frac{g_i}{G}n_i$, while the arrival time $r_i$ is kept unchanged. 
As the actual per-iteration training time $\alpha_i$ of a job can only be obtained after placement, to estimate the job's GPU requirements and its minimum attainable per-iteration training time, we optimistically employ the minimum per-iteration training time $\tilde{\alpha}_i^{\min}$, which is determined in the previous section. Thus, the job duration in instance $A_1$ is calculated as $\frac{g_i}{G}n_i\tilde{\alpha}_i^{\min}$.
Furthermore, since the actual number of job training iterations $n_i$ is unknown at the time of scheduling, we introduce another instance, $\tilde{A}_1$. This instance substitutes $A_1$'s training iteration number, $\frac{g_i}{G}n_i$, with the predicted value, $\frac{g_i}{G}\tilde{n}_i$ (\circled{2}). Consequently, the predicted job duration in $\tilde{A}_1$ is represented as $\frac{g_i}{G}\tilde{n}_i\tilde{\alpha}_i^{\min}$.
We schedule all jobs in $\tilde{A}_1$ first, and order jobs according to their completion time in $\tilde{A}_1$. We then perform job scheduling on the actual cluster following the order (\circled{3}).
By doing so, jobs with larger predicted workloads $\frac{g_i}{G}\tilde{n}_i\tilde{\alpha}^{\min}_i$ are scheduled later due to longer completion times in $\tilde{A}_1$. 
Therefore, the goal of \ouralg{} is to create space for potentially future smaller jobs to be scheduled first, thus reducing the total job completion time.}



\begin{algorithm}[t!]
\DontPrintSemicolon
\SetNoFillComment
\footnotesize
\KwIn{$I, \{S_i\}, \{k_{i,s}\}, g_i, \{p^f_{i,s}, p^b_{i,s}\}, \{d_{i,s}^{in}, d_{i,s}^{out}\}, \{h_{i,s}\}, $\\ $M, g, B_{\texttt{inter}}, B_{\texttt{intra}}$}
\KwOut{$\{t_i, \{x_{i,s}^m\}\}_{i\in [I]}$}
	\While{$t \leq T$}{
		Append completed jobs in $\tilde{A}_1$ using SRPT to \texttt{pending\_queue}\;
		\While {$\texttt{pending\_queue}$ is not empty}{
			$i\leftarrow$ head of \texttt{pending\_queue}\;
			\If {$g_i \leq$ available number of GPUs in the cluster} {
				$\mathcal{M}_i \leftarrow \emptyset$\;
				Pop $i$ from \texttt{pending\_queue}\;
				\If {$\alpha_i^{\max}/\tilde{\alpha}_i^{\min} \geq $ \texttt{COMM\_HEAVY}} {
					Select $g_i$ GPUs from servers with most available GPUs; $\mathcal{M}_i \leftarrow$ these servers\;
					$\{x_{i,s}^m\}\leftarrow$ \heavyedge{}($i, \mathcal{M}_i$)\; 
					Calculate $\alpha_i(\{x_{i,s}^m\})$ using~(\ref{eqn:periteration_time})\;
					\If {$\alpha_i(\{x_{i,s}^m\})/\tilde{\alpha}_i^{\min} \leq $ \texttt{COMM\_HEAVY}} {
						$t_i\leftarrow t$\;
					}
					\Else {
						$\kappa \leftarrow \alpha_i(\{x_{i,s}^m\})$\;
						\For {$t \in \{t+1, \ldots, t+\tau\frac{g_i}{G}\tilde{n}_i\tilde{\alpha}_i^{\min}\}$} {
							Calculate $\{x_{i,s}^m\}$ and $\alpha_i$ based on current server availability\;
							\If {$\alpha_i < \kappa$} {
								$t_i \leftarrow t$; \Break\;
							}
						}
						$t_i \leftarrow t$\;
					}
					
				}
				\Else { 
					Select $g_i$ GPUs from servers with least available GPUs; $\mathcal{M}_i \leftarrow$ these servers\;				
					$\{x_{i,s}^m\}\leftarrow$ \heavyedge{}($i, \mathcal{M}_i$); $t_i\leftarrow t$\;
			}
		}
		\Else {
			$t\leftarrow t+1$\;
		}
		}
	}
	\Return $\{t_i, \{x_{i,s}^m\}\}_{i\in [I]}$\;
\caption{The \ouralg{} Algorithm.}
\label{alg:ouralg}
\end{algorithm}
\setlength{\textfloatsep}{0pt}

{Our \ouralg{} algorithm is detailed in Algorithm~\ref{alg:ouralg}.
The job completion order in $\tilde{A}_1$ is maintained in \texttt{pending\_queue}.}
Let $i$ denote the current head of \texttt{pending\_queue}, \ie, the job to be scheduled. 
If the number of GPUs required by job $i$ (\ie, $g_i$) is less than or equal to the available number of GPUs in the cluster, job $i$ can be scheduled (Line~5), and removed from \texttt{pending\_queue} (Line~7). 
Otherwise, we proceed to the next time-slot (Line~25).

To further improve resource utilization, we classify jobs as either ``communication-heavy'' or ``non-communication-heavy,'' thereby tailoring the scheduling policy to each job's communication pattern.
The {\em rationale} behind this strategy is that communication-heavy jobs have per-iteration training times highly sensitive to GPU mapping due to large communication data sizes, making their worst-case training time $\alpha_i^{\max}$ (with inter-server bandwidth $B_{\texttt{inter}}$) much higher than when allocated to the fewest possible servers.
Jobs are classified by the ratio $\alpha_i^{\max}/\tilde{\alpha}_i^{\min}$. If this ratio exceeds \texttt{COMM\_HEAVY} (1.5 in our experiments), the job is communication-heavy; otherwise, it is non-communication-heavy. Communication-heavy jobs are delayed until sufficient server resources are available, while non-communication-heavy jobs are initiated immediately to maintain workflow efficiency.

For communication-heavy jobs, we prioritize server consolidation (Lines~8--20). We select servers based on maximum availability and calculate $\alpha_i({x_{i,s}^m})$. If $\alpha_i({x_{i,s}^m})/\tilde{\alpha}_i^{\min} \leq \texttt{COMM\_HEAVY}$, we schedule immediately. Otherwise, we delay up to $\tau\frac{g_i}{G}\tilde{n}_i\tilde{\alpha}_i^{\min}$, constantly reassessing allocations for a more efficient configuration, \ie, a lower $\alpha_i$.

For non-communication-heavy jobs, we prioritize immediate execution using a fragmentation-aware strategy (Lines~21–23). Since their per-iteration training times are less affected by placement, we allocate them to servers with lower availability, reserving higher-availability servers for communication-heavy jobs. We then use the \heavyedge{} algorithm for GPU mapping and promptly initiate the job.

\smallskip
\textbf{2) Theoretical Performance Analysis:} 
Let $\Gamma_A$ denote the total job completion time achieved by \ouralg{} for the GPU cluster scheduling problem $A$, and let \opt{} represent the true optimal job completion time. 
Also, let \optAone{} and \optAonepred{} denote the total job completion times of the SRPT-based schedules for instances $A_{1}$ and $\tilde{A}_{1}$, respectively. 

\begin{lemma}
	$\textit{\optAone}  \leq \rho\textit{OPT}_A$, where $\rho = \max_{i\in[I]}\frac{\alpha_i^{\max}}{\alpha_i^{\min}}$.
	\label{lemma:single_multi_server}
\end{lemma}
The detailed proof is given in Appendix~\ref{proof:lemma:single_multi_server}.


\begin{lemma} 
$\Gamma_A$ is no larger than 
{\small
\begin{align*}
(1 +\tau + \frac{\rho G}{G-g^{\max}})\textit{OPT}_{\tilde{A}_1} + I\frac{g^{\max}\alpha^{\max}}{G-g^{\max}}\epsilon + \rho{OPT_A},
\end{align*}}
where $g^{\max} = \max_{i\in[I]}g_i$, and $\alpha^{\max} = \max_{i\in[I]}\alpha^{\max}_i$
	\label{lemma:single_prediction_multi_server}
\end{lemma}
The detailed proof is given in Appendix~\ref{proof:lemma:single_prediction_multi_server}.

\begin{lemma} \label{lemma:single_server_prediction}
$\textit{OPT}_{\tilde{A}_1} \leq \textit{OPT}_{{A_{1}}}  + I\frac{g^{\max}\alpha^{\max}}{G}\epsilon$.
\end{lemma}
The detailed proof is given in Appendix~\ref{proof:lemma:single_server_prediction}.

Then, the total job completion time performance result of \ouralg{} immediately follows from Lemmas~\ref{lemma:single_multi_server}--\ref{lemma:single_server_prediction}:

\begin{theorem}[Total job completion time achieved by \ouralg{}]
	\label{theorem:competitive_ratio}
	$\Gamma_A$ is no larger than 
        {\small
	\begin{displaymath}
		(2 + \tau + \frac{\rho G}{G-g^{\max}})\rho + \frac{2\rho g^{\max}\alpha^{\max}}{\alpha^{\min}}({1+\tau} + \frac{(1+\rho)G}{G-g^{\max}})\bar{\epsilon}
	\end{displaymath}}
times the optimal job completion time $OPT_A$, where $\alpha^{\min} \triangleq \min_{i\in[I]}\alpha^{\min}_i$.
\end{theorem}
\begin{proof}
Combining Lemmas~\ref{lemma:single_multi_server}, \ref{lemma:single_prediction_multi_server} and \ref{lemma:single_server_prediction} yields:

\vspace{-.1in}
{\small
\begin{align*}
&\Gamma_A \!\leq\! (1 \!+\! \tau \!+
\! \frac{\rho G}{G \!-\! g^{\max}})\textit{OPT}_{\tilde{A}_1} \!+
\! I\frac{g^{\max}\alpha^{\max}}{G\!-\!g^{\max}}\epsilon \!+\! \rho{OPT_A} \leq \nonumber\\
& \!\! (2 \!+\! \tau \!+\! \frac{\rho G}{G\!-\!g^{\max}})\rho{OPT_A} \!\!+\! Ig^{\max}\alpha^{\max}\bigg[\frac{1\!\!+\!\tau}{G} \!+\! \frac{1\!+\!\rho}{G\!-\!g^{\max}}\bigg]\epsilon.
\end{align*}}
Assuming each job runs at least one iteration, we have $\rho OPT_A \geq OPT_{A_1} \geq \sum\limits_{i = 1}^{I}(i \times \alpha^{\min})/G = \alpha^{\min}\frac{I(I+1)}{2G}$. 
It then follows that

{\small
\begin{align*}
\!\!\frac{\Gamma_A}{OPT_A} \!<\! \bigg[2 \!+\! \tau \!+\! \frac{\rho G}{G \!-\! g^{\max}}\bigg]\rho \!+\! 
2\rho g^{\max} \bar{\rho} \bigg[1\!+\!\tau \!+\! \frac{(1\!+\!\rho)G}{G\!-\!g^{\max}}\bigg]\bar{\epsilon},
\end{align*}}
where $\bar{\rho} \triangleq \frac{\alpha^{\max}}{\alpha^{\min}}$.
This completes the proof.
\end{proof}
{We remark that our competitive ratio bound is for the worst-case scenario. In this scenario, it is assumed that all jobs could potentially be executed with the maximum per-iteration training time $\alpha^{\texttt{max}}_i$, which rarely happens in practice.
} 
Our numerical evaluations based on real-world data traces and popular DNN models show that the performance of \ouralg{}
is much better than the worst-case competitive ratio bound suggests.
Also, Theorem~\ref{theorem:competitive_ratio} says that the performance of \ouralg{} is closely tied to the average error of the employed prediction model. 
In what follows, we propose an efficient prediction model that provides robust estimates based on the actual characteristics of the {jobs}. 


\smallskip
\textbf{3) Random Forest Based Prediction:}
Studies show that most DDL jobs are recurrent, with nearly 65\% submitted at least five times over two months~\cite{weng2022mlaas}. 
This recurrence provides the opportunity for GPU cluster to perform prediction based on repeated job submissions by applying a hashing function to meta-information (e.g., user details, training dataset, and command-line script), thus generating a unique \texttt{group id} for recurrent jobs.
Leveraging \texttt{group id} and historical job data, we employ random forest regression~\cite{breiman2001random} with mean squared error for tree splitting to predict training iterations based on \texttt{group id} and user id. 
We predict 0 iterations for unseen jobs, treating them as immediately complete in $\tilde{A}_1$ and adding them directly to \texttt{pend\_queue} for swift execution, reducing wait times and enhancing efficiency.
We use 100 trees in our random forest regression.
The high efficiency of forest regression allows frequent retrainings (hourly/daily) for accurate predictions. 
Training with a two-month trace of 700,000 DDLwMP jobs~\cite{weng2022mlaas} takes only 80 seconds. Combined with \ouralg{}, our prediction model enables efficient resource allocation and job scheduling in GPU clusters.


\vspace{-.05in}
\section{Performance Evaluation}
\label{sec::eval}

In this section, we conduct both {\em real-world} data-trace-driven testbed experiments and simulation studies to evaluate the performance and efficacy of our proposed \ouralg{} algorithm.

\subsection{Real-World GPU Cluster Testbed Experiments}

{\bf 1) System Settings:}
\noindent\textit{1-a) Implementation and Testbed:}
We implement \ouralg{} using Python and PyTorch 2.1.1~\cite{paszke2019pytorch} with 4634 lines of code. The evaluation of \ouralg{} is conducted on a single server equipped with two NVIDIA H100 NVL GPUs. To simulate a GPU cluster, we utilize the Multi-Instance GPU (MIG)~\cite{mig} technique, partitioning the two H100 GPUs into 14 virtual GPUs (vGPUs), each with 12 GB of GPU memory. The scheduling overhead per job is within 5s.
Due to the MIG configurations, inter-vGPU communication is limited to the PCIe bandwidth of 128 GB/s. Consequently, GPU mapping does not significantly impact our testbed experiment. Therefore, we set the delay factor to zero in \ouralg{}. 
For more heterogeneous inter-GPU networks, we evaluate the performance of \ouralg{} in the simulation studies later in this section.

\textit{1-b) Deep Learning Workload:} The dataset for our job analysis is obtained from an open-source two-month deep learning workload trace collected from a production cluster with 6000 GPUs~\cite{weng2022mlaas}. 
This data-trace contains features including job duration, submission time, user id of the individual submitting the job, requested number of GPUs, and \texttt{group id} that identifies recurring jobs.
After completing a data cleaning process, we obtain a total of 758,223 jobs for analysis. 

However, this data-trace does not provide the training jobs' DNN model information. 
To address this issue, we profile nine representative DNN models on the vGPUs: three image classification models on the ImageNet dataset~\cite{deng2009imagenet} and six natural language processing (NLP) models.
The details of this model profiling are summarized in Table~\ref{table:model}.
Here, BERT-large and XLNet-large are profiled on the SQuAD2.0 dataset~\cite{rajpurkar2018know}. 
For T5 and the three versions of GPT models that cannot be accommodated on a single GPU, we construct a smaller model consisting of three layers from the original model, which will be used for profiling with a token sequence length of 512.
{The distributed training configurations for each model are derived from the planner proposed in~\cite{luo2022efficient}, which calculates multiple configurations per model.}
We assign each model and the derived distributed training configuration to a job group (\ie, a group of recurrent jobs) following the GPU requirement in the trace.
If a job in a group requires only a single GPU, we pair the group with a model with a single-GPU training configuration. 
Otherwise, if the job group demands more than one GPU, we randomly select a model and one of its training configurations for the group. 
{The number of job training iterations} is computed by dividing the job duration in the trace by its approximate minimum per-iteration training time, $\tilde{\alpha}_i^{\min}$.

Due to the limited size of our local testbed, we randomly selected three sets of 75 consecutive jobs from the original traces. We uniformly scaled down the job arrival times and training iterations to 10\% of the original data, resulting in a scheduling period on the order of hours per method.

\begin{figure}[!t]
  \centering
  \includegraphics[width=0.65\columnwidth]{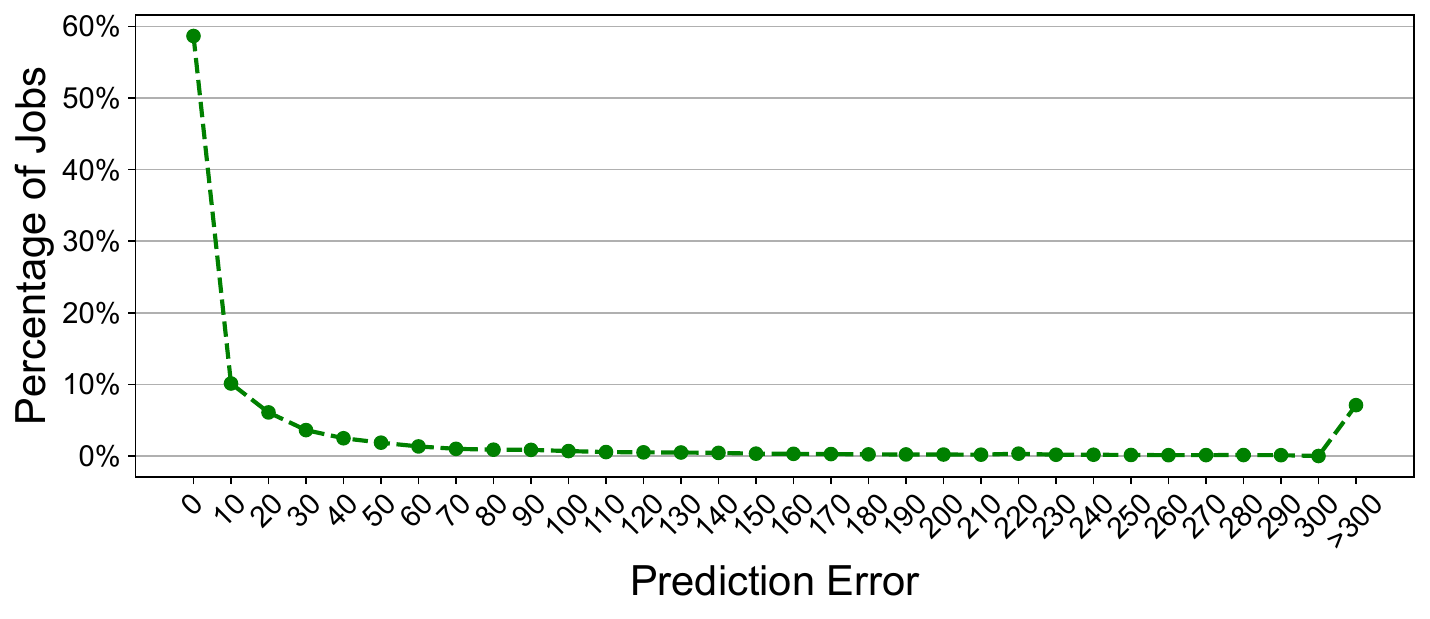}
  \caption{Percentage of jobs: different prediction errors.}
  \label{fig:prediction_error}
\end{figure}
\setlength{\textfloatsep}{0pt}

{\scriptsize
\begin{table}[!t]
\centering
\begin{tabular}{|l|c|c|}
\hline
Model & \# of Parameters & Mini-Batch Size \\ \hline
VGG19~\cite{simonyan2014very} & 144M & 32 \\ \hline
ResNet152~\cite{he2016deep} & 60M & 4 \\ \hline
Inception-V3~\cite{szegedy2016rethinking} & 24M & 32 \\ \hline
BERT-large~\cite{devlin2018bert} & 340M & 4 \\ \hline
XLNet-large~\cite{yang2019xlnet} & 550M & 4 \\ \hline
T5~\cite{raffel2020exploring} & 11B & 8 \\ \hline
GPT~\cite{brown2020language} & 6.7B/13B/175B & 32/32/16 \\ \hline
\end{tabular}
\caption{DNN models.}
\label{table:model}
\end{table}}
\setlength{\textfloatsep}{0pt}

\begin{figure}[!t]
  \centering
  \includegraphics[width=0.5\columnwidth]{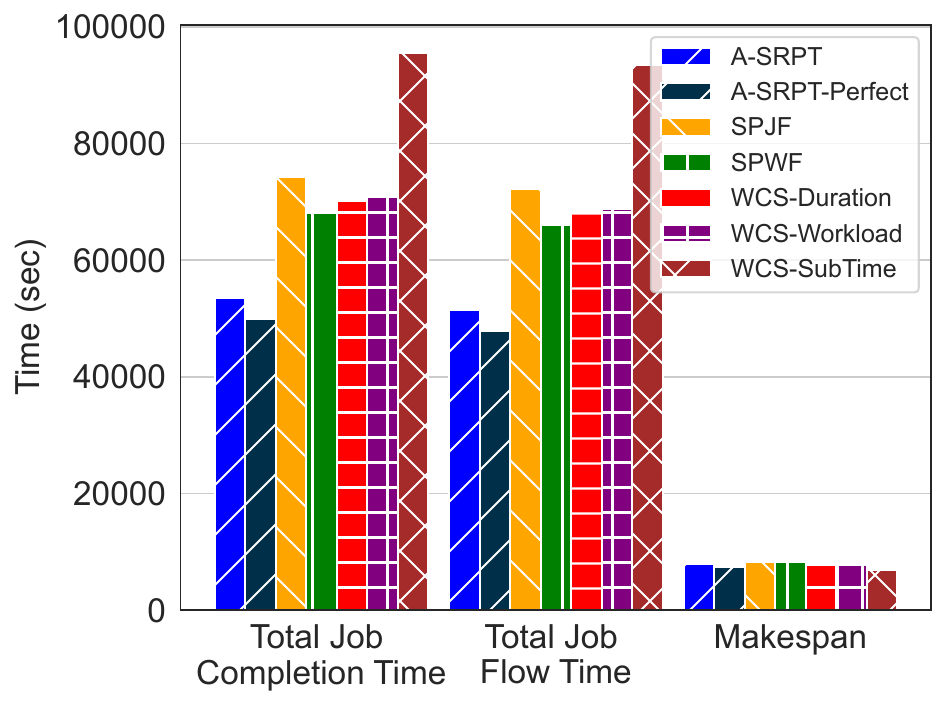}
  \caption{Testbed experiment performance.}
  \label{fig:testbedperf}
\end{figure}
\setlength{\textfloatsep}{0pt}

\textit{1-c) Prediction Model:} 
We use the first 80\% jobs in the trace to train our random forest regression prediction model, completing in just 84 seconds. 
The prediction error is depicted in Fig.~\ref{fig:prediction_error}, which shows that approximately 60\% of the jobs are predicted correctly.
Although there remains a non-negligible prediction error in a small fraction of jobs, our subsequent evaluation reveals that our algorithm outperforms the baseline performance even with imperfect predictions.

\def\spjf{\textit{SPJF}}
\def\spwf{\textit{SPWF}}
\def\wcsduration{\textit{WCS-Duration}}
\def\wcsworkload{\textit{WCS-Workload}}
\def\wcssub{\textit{WCS-SubTime}}
\def\asrptperfect{\textit{A-SRPT-Perfect}}

{\em 1-d) Baselines:} 
Our \ouralg{} algorithm is compared with five baseline GPU cluster scheduling algorithms:
{(1) \spjf{} {\em (Shortest Predicted Job First):} This approach schedules jobs based on their predicted durations as proposed by MLaaS~\cite{weng2022mlaas};
(2) \spwf{} {\em (Shortest Predicted Workload First):} This policy proposed in Tiresias~\cite{gu2019tiresias} schedules jobs according to the product of predicted durations and the number of required GPUs;}
(3) \wcsduration{} {\em (Work-Conserving Scheduler, WCS~\cite{zheng2015exploiting} by Duration):} This approach continuously schedules jobs to use available GPUs within the cluster following the order based on their predicted duration;
(4) \wcsworkload{}: Variant of (3), sequencing by predicted workload;
(5) \wcssub{}: Variant of (3), arranged by submission time.
All baselines adopt the \heavyedge{} algorithm for GPU mapping in both testbed and simulation experiments.

{\bf 2) Experimental Results:}
We present the real testbed results in Fig.~\ref{fig:testbedperf}, averaged over three job sets. 
The total job flow time is defined as the difference between each job's completion time and arrival time, and the makespan is the completion time of the final job. 
We include the baseline \asrptperfect{}, which uses \ouralg{} with perfect knowledge of job durations (\ie, perfect prediction). 
Our \ouralg{} achieves performance close to \asrptperfect{}, with only 7\% longer total job completion time, and significantly outperforms all other baselines. 
While \textit{WCS} baselines achieve shorter system makespans, they prioritize scheduling longer training jobs whenever possible. This blocks the timely execution of later arriving shorter jobs, resulting in larger total job completion times. In contrast, our algorithm reduces the total job completion time by up to 44\%.

\subsection{Large-Scale Simulations}
{\bf 1) System Settings:}
\noindent\textit{1-a) System Settings:} 
We consider a cluster consisting of 250 servers, each equipped with eight GPUs. The NIC bandwidth of each server is set to 10Gbps, and the inter-GPU communication bandwidth within each server is 300GB/s, based on the NVLink specs of NVIDIA V100 GPUs. We profiled all DNN models on a single NVIDIA V100 GPU. For scheduling, we randomly sample consecutive jobs from the original trace.

\begin{figure*}[!t]
        \centering
	\begin{minipage}[t]{0.23\textwidth}
		\includegraphics[width=\textwidth]{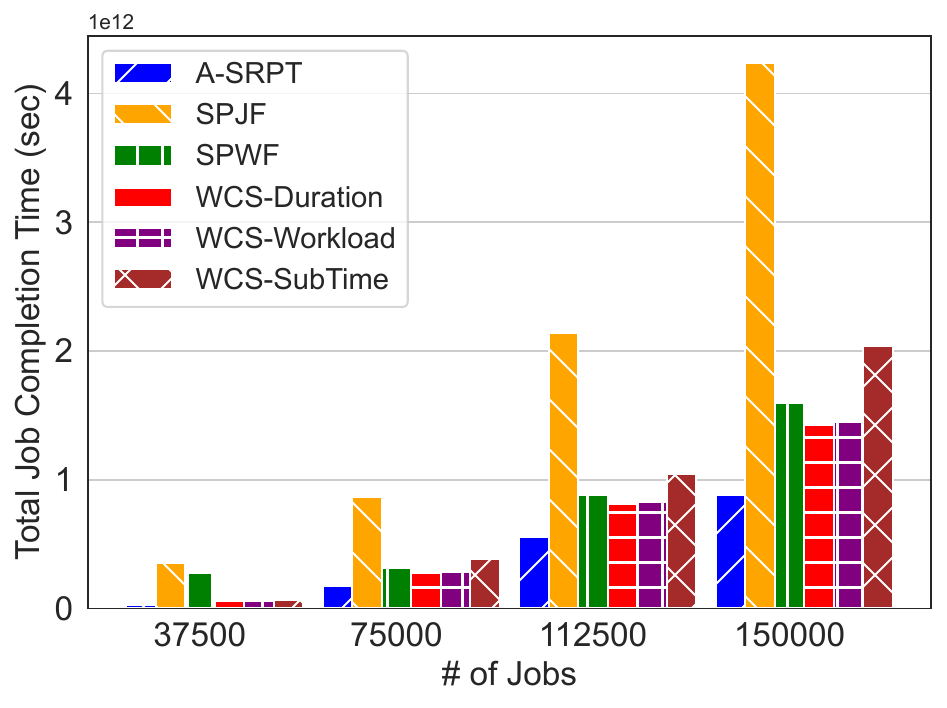}
		\caption{Total job completion time comparisons with different numbers of jobs.}
		\label{fig:diffnumjobs}
	\end{minipage}
	\begin{minipage}[t]{0.23\textwidth}
		\includegraphics[width=\textwidth]{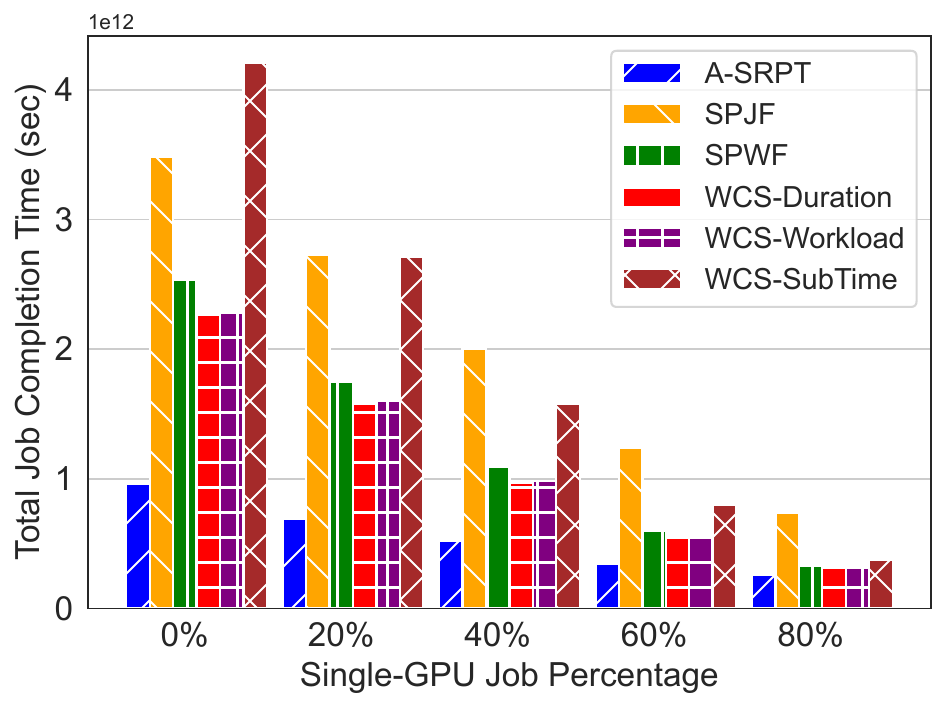}
		\caption{Completion time comparisons with different percentages of single-GPU jobs.}
		\label{fig:diffsinglegpujobs}
	\end{minipage}
	\begin{minipage}[t]{0.23\textwidth}
		\includegraphics[width=\textwidth]{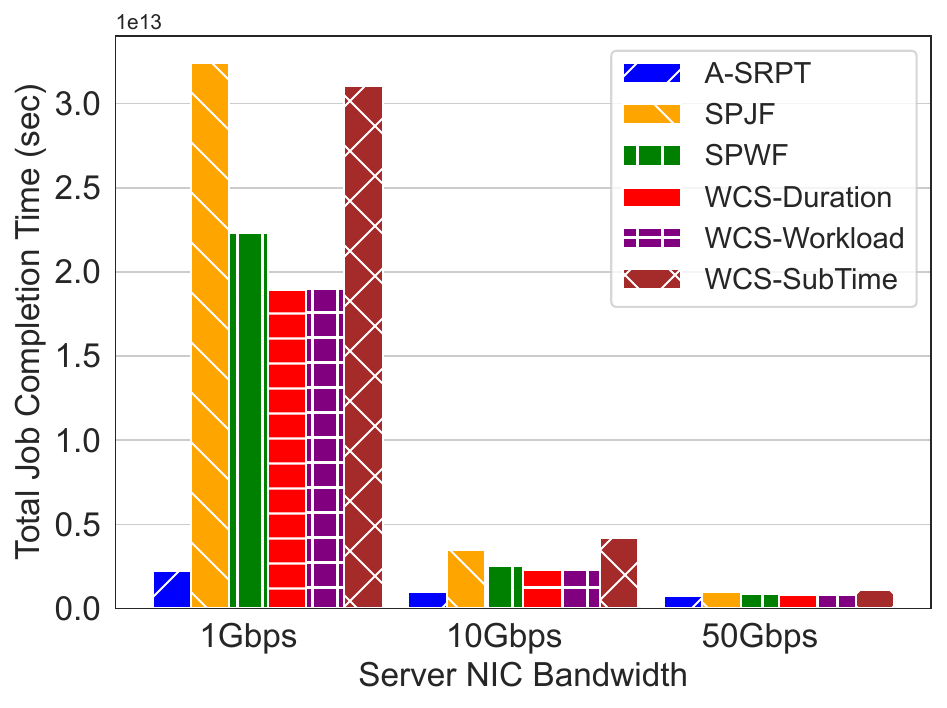}
		\caption{Total job completion time with different server NIC bandwidths.}
		\label{fig:diffbandwidth}
	\end{minipage}
	\begin{minipage}[t]{0.23\textwidth}
            \includegraphics[width=0.94\textwidth]{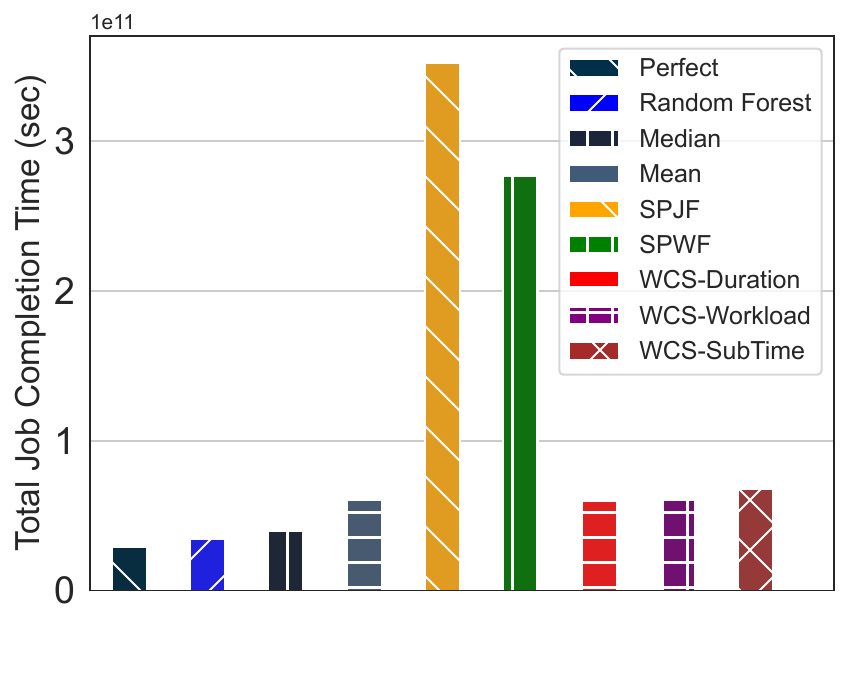}
		\caption{Total completion time comparisons: different prediction models and baselines.}
		\label{fig:diffprediction}
 	\end{minipage}
   \vspace{-.1in}
\end{figure*}

{\bf 2) Experimental Results:}
\noindent\textit{2-a) Different Number of Jobs:}
As the number of jobs increases, the workload and job diversity grow, challenging the online algorithm’s ability to handle varying job sizes.
Fig.~\ref{fig:diffnumjobs} shows total job completion times for \ouralg{} and baselines with job counts from 37,500 to 150,000 (5\% to 20\% of the trace).
\spjf{} performs the worst due to its rigid strategy based solely on predicted durations, neglecting varying GPU demands. If the shortest job does not fit, it will not schedule longer jobs with fewer GPU demands. 
\spwf{} balances job duration with GPU needs, leading to better workload distribution. 
\wcsduration{} and \wcsworkload{} enhance GPU utilization but delay larger jobs by prioritizing smaller ones. \ouralg{} consistently outperforms baselines, reducing total job completion times by 31\% to 91\%.

\textit{2-b) Different Percentages of Single-GPU Jobs:} The original trace~\cite{weng2022mlaas} has over 70\% single-GPU jobs, making scheduling less challenging due to minimal server assignment. 
Thus, we fix the number of jobs at 75,000 and vary the percentage of single-GPU jobs, with jobs randomly set for single-GPU or distributed training. 
As the fraction of distributed jobs increases, the scheduling problem becomes harder due to higher workloads and complex communication. 
Fig.~\ref{fig:diffsinglegpujobs} shows that as single-GPU jobs decrease from 80\% to 0\%, \ouralg{} increasingly outperforms baselines, reducing total job completion time by 16\% to 57\%.

{\em 2-c) Different Server NIC Bandwidths:}
{We evaluate \ouralg{} with server NIC bandwidths from 1 Gbps to 50 Gbps, using the job set with 0\% single-GPU jobs.} 
{Lower bandwidth exacerbates communication overhead, yielding longer total job completion times under poor scheduling. 
Fig.~\ref{fig:diffbandwidth} shows \ouralg{} maintains consistent performance gains, while baselines falter at 1 Gbps.
Notably, at 50 Gbps, \ouralg{} outperforms the best baseline \wcsduration{} by 12\%, and at 1 Gbps, it reduces total job completion time by up to 92\%, demonstrating its effectiveness in handling communication overhead and ensuring efficient job training.}
{\em 2-d) Different Prediction Models:}
We now examine the performance of our prediction model in Fig.~\ref{fig:diffprediction} using jobs with GPU demands following the original trace.
Our random forest regression model is compared with simpler methods based on the mean and median of previous job iterations, as well as a perfect prediction model (\ie, \asrptperfect{}).
All other baselines use random forest regression. 
The average errors for the random forest, median-based, and mean-based models are 369, 563, and 593, respectively. 
The random forest model outperforms simpler methods due to lower average error and is only 14\% less efficient than the perfect model, while less accurate models (e.g., mean-based) significantly degrade algorithm performance.

{\tiny
\begin{table}[t!]
    \centering
    \begin{tabular}{|c|cc|cc|}
        \hline
        \multirow{2}{*}{\textbf{Model}} & \multicolumn{2}{c|}{\heavyedge{}} & \multicolumn{2}{c|}{\textsf{ILP}} \\
        \cline{2-5}
        & \textbf{PITT (ms)} & \textbf{PCT (ms)} & \textbf{PITT (ms)} & \textbf{PCT (ms)} \\
        \hline
        VGG19 & 88.11 & 1.94 & 82.96 & 55318.86 \\
        GPT-175B & 10.14 & 1.52 & 10.14 & 2288.12 \\
        \hline
    \end{tabular}
    \caption{Per-iteration training time (PITT) and placement computation time (PCT): \heavyedge{} vs. \textsf{ILP}}
    \label{tab:comparison_placement}
\end{table}}

{\em 2-e) \heavyedge{} vs. Integer Linear Programming (ILP):}
Finally, we evaluate the performance of \heavyedge{}, with results shown in Table~\ref{tab:comparison_placement}. 
In comparison, the placement is formulated as an ILP problem based on~\cite{archer2023pipeline} and solved optimally using the Gurobi Optimizer~\cite{gurobi}.
Experiments were conducted on a MacBook Pro (M1 MAX chip, 64 GB memory). 
We compare the per-iteration training time (PITT) and placement computation time (PCT) for two of our profiled models, 
averaging results over 20 cases with varying GPU availability per server.
For the VGG19 model, the heterogeneity in computation time and data communication presents challenges in GPU mapping. 
\heavyedge{} achieves a PITT only 6\% longer than the optimal ILP solution, while computing in under TWO milliseconds compared to ILP’s 55+ seconds. 
Moreover, for the GPT-175B model, the uniform structure allows \heavyedge{} to find a solution 1500 times faster than the ILP.

\section{{Discussions}}
\label{sec::discussion}
{We note that the landscape of parallelism for distributed deep learning training continues to evolve.
New methods, such as tensor parallelism~\cite{shoeybi2019megatron} and expert parallelism~\cite{hwang2023tutel}, have been key enablers for training extremely large-scale foundation models~\cite{bommasani2021opportunities}. Reflecting on this, it is interesting to discuss how \ouralg{} can be extended to work with emerging parallelisms to enable DDL scheduling designs for the future.
}

\smallskip
\noindent $\triangleright$ \textbf{Tensor parallelism.} 
Tensor parallelism splits layers across multiple GPUs, necessitating extensive inter-GPU communication through AllReduce operations~\cite{shoeybi2019megatron}. 
To adapt \heavyedge{} for tensor parallelism, we modify our graph model $\Omega = (\mathcal{V}, \mathcal{E})$ to represent tensor slices as vertices and AllReduce operations as weighted edges.
For communication efficiency, all tensor slices of a layer must reside within a single server. 
To accommodate this, \ouralg{} delays the start of tensor parallelism jobs until sufficient server capacity is available.

\smallskip
\noindent $\triangleright$ \textbf{Expert parallelism.} 
Expert parallelism in Mixture-of-Experts (MoE) models distributes different `expert' layers across GPUs, posing challenges in balancing workloads and managing inter-GPU communication~\cite{hwang2023tutel}. 
We can represent expert groups as vertices in our graph-based model. Communication, characterized by sparse activations/gradients and token routing, is represented as weighted edges. 
Due to the dynamic data transfer patterns presented in MoE training, we can set the edge weights based on the estimated average communication costs. 
This allows MoE jobs to be integrated into our unified graph model, enabling effective placement and scheduling with \heavyedge{} and \ouralg{}.

\section{Conclusion}
\label{sec::conclusion}
In this paper, we investigated online scheduling for distributed deep learning with mixed parallelism (DDLwMP) jobs in GPU clusters. 
We introduced the adaptive shortest-remaining-processing-time first (\ouralg{}) scheduling method, which integrates: 1) a GPU mapping algorithm that strategically assigns GPUs to job stages to minimize communication overhead by co-locating communication-intensive parts, and 2) an online scheduling algorithm that uses a prediction model for job scheduling. 
By modeling each DDL job as a graph, our GPU mapping algorithm reduces communication overhead effectively. 
Additionally, we proposed an online scheduling algorithm that transforms the complex GPU cluster scheduling problem into a single-machine instance, which can be optimally solved. 
The scheduling decisions from this simplified problem then guide the actual GPU cluster scheduling.
Theoretical analysis and trace-driven experiments demonstrated \ouralg{}’s efficacy, achieving up to 92\% reduction in total job completion time compared to baselines.

\bibliographystyle{IEEEtran}
\bibliography{./main}

\appendices
\section{Proof of Lemma~\ref{lemma:single_multi_server}}
\label{proof:lemma:single_multi_server}
\begin{proof}Let $\{c_i^*\}$ be the completion times in an optimal schedule of $A$, so $\textit{OPT}_A = \sum_i c_i^*$. 
Create a new instance by setting each job's per-iteration time to $\tilde{\alpha}_i^{\min} \le \rho\,\alpha_i^{\min}$, and scale the timeline of the optimal schedule by $\rho$. Since $\tilde{\alpha}_i^{\min}/\alpha_i^{\min} \le \alpha_i^{\max}/\alpha_i^{\min} \le \rho$, 
this scaled schedule remains valid and each job now completes by time $\rho\,c_i^*$, giving a total completion time at most $\rho\,\textit{OPT}_A$.

We can now construct a schedule of $A_1$ based on the scaled schedule.
To construct a schedule for $A_1$, observe that at any time $t$ in the scaled schedule if $I_t$ is the set of jobs running (each using $g_i$ GPUs), we instead assign each job $i \in I_t$ a fraction $\frac{g_i}{G}$ of the total GPU capacity. This does not delay any completion times, so the total completion time of such schedule of $A_1$ remains at most $\rho\textit{OPT}_A$.
Hence, $\textit{\optAone} \le \rho\textit{OPT}_A$.
\end{proof}

\section{Proof of Lemma~\ref{lemma:single_prediction_multi_server}}
\label{proof:lemma:single_prediction_multi_server}
\begin{proof}
Without loss of generality, let job $i \in [I]$ be the $i$-th job completed in $\tilde{A}_{1}$, and let $C_i(\textit{OPT}_{\tilde{A}_{1}})$ be its completion time in that schedule.
Then
{\small
\begin{equation*}
	C_i(\textit{OPT}_{\tilde{A}_1}) \geq \sum\limits_{j=1}^i\frac{g_j}{G}\tilde{n}_j\tilde{\alpha}^{\min}_j
\end{equation*}}

Let $C_i(\Gamma_A)$ be the completion time of job $i$ under \ouralg{}. Consider a worst-case scenario where jobs $1,\dots,i-1$ begin only after $C_i(\textit{OPT}_{\tilde{A}_1})$. In that scenario, the time to complete these $i-1$ jobs satisfies

{\small
\begin{eqnarray*}
	\mathrm{makespan}_{i-1}\leq C_i(\textit{OPT}_{\tilde{A}_1}) + \sum\limits_{j=1}^{i-1}\frac{g_jn_j\alpha^{\max}_j}{G-g^{\max}} + \sum\limits_{j=1}^{i-1}\frac{\tau g_j\tilde{n}_j\tilde{\alpha}^{\min}_j}{G}
\end{eqnarray*}}
In the above, the second term arises because with no further delays, we can ensure at least $G - g^{\max}$ GPUs are continuously busy for each job's training, and the third term reflects additional delay for communication-heavy jobs, as each delay increases the makespan by up to $\frac{\tau g_j\tilde{n}_j\tilde{\alpha}^{\min}_j}{G}$.

Let $\mathcal{U}$ denote the set of jobs that are underestimated ($\tilde{n}_i< n_i$), and $\mathcal{O}$ denote the set of jobs that are overestimated. A straightforward bounding argument yields

{\small
\begin{eqnarray}
	\sum\limits_{j=1}^{i-1}\frac{g_jn_j\alpha^{\max}_j}{G-g^{\max}} & = & \sum\limits_{j\in \mathcal{U}\cap [i-1]}\frac{g_j}{G-g^{\max}}(\tilde{n}_j + \epsilon_j)\alpha^{\max}_j +\nonumber \\
	&& \sum\limits_{j\in \mathcal{O}\cap [i-1]}\frac{g_j}{G-g^{\max}}(\tilde{n}_j - \epsilon_j)\alpha^{\max}_j + \nonumber\\
	&& \sum\limits_{j\in [i-1]-\mathcal{U}-\mathcal{O}}\frac{g_j}{G-g^{\max}}\tilde{n}_j\alpha^{\max}_j\nonumber\\
	&\leq & \sum\limits_{j=1}^{i-1}\frac{g_j\tilde{n}_j\alpha^{\max}_j}{G-g^{\max}} + \sum\limits_{j\in [i-1]}\frac{g_j\epsilon_j\alpha^{\max}_j}{G-g^{\max}} \nonumber\\
	&\leq & \frac{\rho G}{G-g^{\max}} C_i(\textit{OPT}_{\tilde{A_{1}}}) + \frac{g^{\max}\alpha^{\max}}{G-g^{\max}}\epsilon 
\end{eqnarray}}

Hence,
{\small
\begin{eqnarray}
	\mathrm{makespan}_{i-1}&\leq& (1+\frac{\rho G}{G-g^{\max}})C_i(\textit{OPT}_{\tilde{A}_1}) + \frac{g^{\max}\alpha^{\max}}{G-g^{\max}}\epsilon +\nonumber\\
	& & \sum\limits_{j=1}^{i-1}\frac{\tau g_j\tilde{n}_j\tilde{\alpha}^{\min}_j}{G}	\label{eqn:makespan_bound}
\end{eqnarray}}

Since $C_i(\Gamma_A)$ is at most $\mathrm{makespan}_{i-1}$ plus the time for job $i$ itself, we get
{\small
\begin{eqnarray}
	C_i(\Gamma_{A}) & \leq & \mathrm{makespan}_{i-1} + \frac{\tau g_i\tilde{n}_i\tilde{\alpha}^{\min}_i}{G} + n_i\alpha^{\max}_i\nonumber\\
	& \leq & (1 + \tau + \frac{\rho G}{G-g^{\max}})C_i(\textit{OPT}_{\tilde{A}_1})+\frac{g^{\max}\alpha^{\max}}{G-g^{\max}}\epsilon \nonumber\\
	&&+ n_i\alpha^{\max}_i
\end{eqnarray}}

Observing that $\sum\limits_{i\in[I]}n_i\alpha^{\max}_i \leq \rho \sum\limits_{i\in[I]}n_i\alpha^{\min}_i \leq \rho OPT_A$, we sum over all $i \in [I]$ to derive:
{\small
\begin{eqnarray*}
	\Gamma_{A} & \triangleq & \sum\limits_{i\in [I]}C_i(\Gamma_{A})  \nonumber\\
	& \leq & (1 + \tau + \frac{\rho G}{G-g^{\max}})\textit{OPT}_{\tilde{A}_1} + I\frac{g^{\max}\alpha^{\max}}{G-g^{\max}}\epsilon + \rho{OPT_A} \nonumber\\
\end{eqnarray*}
}
\end{proof}

\section{Proof of Lemma~\ref{lemma:single_server_prediction}}
\label{proof:lemma:single_server_prediction}
\begin{proof}

We adapt the analytic framework of~\cite{bampis2022scheduling} by building an auxiliary schedule (\aux{}) on a modified instance $A^{aux}_{1}$. The idea is to transform both the schedule \optAone{} to \optAonepred{} and the instance $A_{1}$ to $\tilde{A_{1}}$. 
For simplicity, we abuse \aux{} to denote the objective value of the schedules respectively. We achieve the transformation in two phases, one to bound the overestimated jobs ($\mathcal{O}$), and the other to bound the underestimated jobs ($\mathcal{U}$).

\noindent\textbf{Phase 1: Handle Overestimated Jobs ($\mathcal{O}$).}
We begin by bounding all jobs $i \in \mathcal{O}$, that is, all overestimated jobs. We create $A^{aux}_{1}$ by replacing the execution training iterations of any overestimated job $i$ in $A_{1}$, denoted as $n_i$, with the predicted execution time $\tilde{n}_i$. Let \aux{} represent the SRPT schedule (i.e., the optimal schedule) on $A^{aux}_{1}$. We then modify \optAone{} to be a schedule on $A^{aux}_{1}$, which we denote as $\bar{\textit{OPT}}_{A_1}$. 
We define the \textit{execution part} of a job as the uninterrupted period during which the job runs.
$\bar{\textit{OPT}}_{A_1}$ is initialized as \optAone{}. We modify $\bar{\textit{OPT}}_{A_1}$ by iterating over the schedule of all overestimated jobs in \optAone{}. For each job $i \in \mathcal{O}$, we extend the duration of its final execution part in ${\textit{OPT}}_{A_1}$ by $\frac{g_i}{G}\epsilon_i\tilde{\alpha}^{\min}_i$ to obtain $\bar{\textit{OPT}}_{A_1}$. All subsequent execution parts will be delayed by $\frac{g_i}{G}\epsilon_i\tilde{\alpha}^{\min}_i$, increasing $\bar{\textit{OPT}}_{A_1}$ by up to $I\frac{g_i}{G}\epsilon_i\tilde{\alpha}^{\min}_i$. By processing all jobs in $\mathcal{O}$, we increase $\bar{\textit{OPT}}_{A_1}$ by up to $I\sum\limits_{i \in \mathcal{O}}\frac{g_i}{G}\epsilon_i\tilde{\alpha}^{\min}_i$. We denote the total job completion time of \aux{} at the end of phase 1 as $\textit{AUX}[1]$. Therefore, we have:

{\small
\begin{equation}
	\textit{AUX}[1] \leq \bar{\textit{OPT}}_{A_1} \leq \textit{OPT}_{{A_{1}}} + I\sum\limits_{i\in \mathcal{O}}\frac{g_i}{G}\epsilon_i\tilde{\alpha}^{\min}_i
	\label{eqn:phase1}
\end{equation}}

\noindent\textbf{Phase 2: Handle Underestimated Jobs ($\mathcal{U}$).} We now shift our focus to the jobs in $\mathcal{U}$, \ie, the underestimated jobs, transforming \aux{} from \textit{AUX}[1] into \optAonepred{}. We carefully compare the schedules of \aux{} and \optAonepred{} from time 0 until the first time point $t$ where \aux{} and \optAonepred{} schedule two distinct jobs, denoted as $i_{\mathrm{aux}}$ and $i_{\mathrm{pred}}$ respectively.
Considering time $t$, as $\textit{AUX}[1]$ chooses $i_{\mathrm{aux}}$ rather than $i_{\mathrm{pred}}$, we have: 

{\small
\begin{equation}
	\frac{g_{i_{\mathrm{pred}}}}{G}\tilde{n}_{i_{\mathrm{pred}}}\tilde{\alpha}^{\min}_{i_{\mathrm{pred}}} < \frac{g_{i_{\mathrm{aux}}}}{G}\tilde{n}_{i_{\mathrm{aux}}}\tilde{\alpha}^{\min}_{i_{\mathrm{aux}}}< \frac{g_{i_{\mathrm{pred}}}}{G}{n}_{i_{\mathrm{pred}}}\tilde{\alpha}^{\min}_{i_{\mathrm{pred}}}
	\label{eqn:i_aux_vs_i_pred}
\end{equation}}
as $i_{\mathrm{pred}}$ belongs to $\mathcal{U}$. We then decrease the training iteration of $i_{\mathrm{pred}}$ in $A^{aux}_{1}$ from ${n}_{i_{\mathrm{pred}}}$ to $\tilde{n}_{i_{\mathrm{pred}}}$. We maintain the schedule \aux{} up to $t$, and employ SRPT to schedule the remaining workload in $A^{aux}_{1}$, resulting in the updated schedule \aux{}. This approach ensures that job $i_{\mathrm{pred}}$ will be scheduled starting from $t$ in \aux{}, and reduces \aux{} by at least $\frac{g_{i_\mathrm{pred}}}{G}\epsilon_{i_\mathrm{pred}}\tilde{\alpha}^{\min}_{i_\mathrm{pred}}$. We then proceed to identify the subsequent time step $t'$ where \aux{} and \optAonepred{} diverge, and repeat the process. By reducing the actual duration of every underestimated job to the predicted duration, we transform \aux{} to \optAonepred{}. Let $\textit{AUX}[2]$ denote the objective value of \aux{} at the conclusion of phase 2. We now have:

{\small
\begin{equation}
	\textit{OPT}_{\tilde{A_{1}}} = \textit{AUX}[2] \leq \textit{AUX}[1] - \sum\limits_{i\in \mathcal{U}}\frac{g_i}{G}\epsilon_i\tilde{\alpha}^{\min}_i
	\label{eqn:phase2}
\end{equation}
}

Combining (\ref{eqn:phase1}) and (\ref{eqn:phase2}), we have:

{\small
\begin{eqnarray*}
	\textit{OPT}_{\tilde{A_{1}}} & \leq & \textit{AUX}[1] - \sum\limits_{i\in \mathcal{U}}\frac{g_i}{G}\epsilon_i\tilde{\alpha}^{\min}_i \nonumber\\
	&\leq& \textit{OPT}_{{A_{1}}}  + I\sum\limits_{i\in \mathcal{O}}\frac{g_i}{G}\epsilon_i\tilde{\alpha}^{\min}_i - \sum\limits_{i\in \mathcal{U}}\frac{g_i}{G}\epsilon_i\tilde{\alpha}^{\min}_i \nonumber\\
	& \leq & \textit{OPT}_{{A_{1}}}  + I\frac{g^{\max}}{G}\epsilon\alpha^{\max} \nonumber\\
\end{eqnarray*}}
\end{proof}

\end{document}